\documentclass[journal,12pt,onecolumn,draftclsnofoot,]{IEEEtran}
\IEEEoverridecommandlockouts
\usepackage{caption}
\usepackage{subfigure}
\usepackage{amsthm}
\newtheorem{Theo}{Theorem}

\usepackage{amsmath}
\usepackage{mathtools}
\usepackage{url}
\usepackage{algorithm}
\usepackage{algorithmic}
\usepackage{dsfont}
\usepackage{bm}

\usepackage{cite}
\usepackage{amsmath,amssymb,amsfonts}
\usepackage{algorithmic}
\usepackage{graphicx}
\usepackage{textcomp}
\usepackage{xcolor}
\def\BibTeX{{\rm B\kern-.05em{\sc i\kern-.025em b}\kern-.08em
    T\kern-.1667em\lower.7ex\hbox{E}\kern-.125emX}}
\begin{document}

\title{Cooperative Service Caching and Workload Scheduling in Mobile Edge Computing
}
\author{Xiao~Ma,~\IEEEmembership{Member,~IEEE,}
        Ao~Zhou,~\IEEEmembership{Member,~IEEE,}
        Shan~Zhang,~\IEEEmembership{Member,~IEEE,}
        Shangguang Wang,~\IEEEmembership{Senior~Member,~IEEE,}
\thanks{Xiao Ma, Ao Zhou and Shangguang Wang are with the State Key Laboratory of Networking and Switching Technology, Beijing University of Posts and Telecommunications, Beijing, China, 100876.
E-mail: maxiao18@bupt.edu.cn, aozhou@bupt.edu.cn, sgwang@bupt.edu.cn.}
\thanks{Shan Zhang is with the School of Computer Science and Engineering, Beihang University, Beijing, China, 100191. 	
E-mail: zhangshan18@buaa.edu.cn.}
}

\maketitle

\begin{abstract}
Mobile edge computing is beneficial to reduce service response time and core network traffic by pushing cloud functionalities to network edge.
Equipped with storage and computation capacities, edge nodes can cache  services of resource-intensive and delay-sensitive mobile applications and process the corresponding computation tasks without outsourcing to central clouds.
However, the heterogeneity of edge resource capacities and inconsistence of edge storage and computation capacities make it difficult to jointly fully utilize the storage and computation capacities when there is no cooperation among edge nodes.
To address this issue, we consider cooperation among edge nodes and investigate cooperative service caching and workload scheduling in mobile edge computing.
This problem can be formulated as a mixed integer nonlinear programming problem, which has non-polynomial computation complexity.
To overcome the challenges of subproblem coupling, computation-communication tradeoff, and edge node heterogeneity, we develop an iterative algorithm called ICE.
This algorithm is designed based on Gibbs sampling, which has provably near-optimal results, and the idea of water filling, which has polynomial computation complexity.
Simulations are conducted and the results demonstrate that our algorithm can jointly reduce the service response time and the outsourcing traffic compared with the benchmark algorithms.
\end{abstract}


\section{Introduction}
The proliferation of mobile devices and the advancement of Internet of things are promoting the emergence of resource-intensive and delay-sensitive mobile applications, such as objective recognition, augmented reality, and mobile gaming.
Mobile cloud computing proposes to offload these applications to central clouds, which, however, suffers from the uncontrolled wide area network delay and is hard to guarantee the quality of service of delay-sensitive applications \cite{Cuervo2010MAUI,Chun2011clonecloud,satyanarayanan2009case}.
Moreover, according to the prediction of Cisco, the growth rate of mobile data required to be processed will far exceed the capacity of central clouds in 2021 \cite{Cisco2018cloud}.
Limiting the outsourcing traffic to central clouds becomes a critical concern of network operators.
Mobile edge computing has emerged as a promising solution to addressing above concerns \cite{ETSI2016edge,ha2013just}.
A typical form of mobile edge computing is to endow mobile base stations (also named as edge nodes) with cloud-like functions by deploying storage and computation capacities distributedly.
Through caching the services (including the program codes and the related databases) of mobile applications at edge nodes, mobile edge computing is able to process the corresponding computation tasks at network edge, benefiting from the reduced service response time and outsourcing traffic to central clouds.

Compared with mobile cloud computing which has elastic resource capacity, the main limitation of mobile edge computing is the limited resource capacities of edge nodes.
When there is no cooperation among edge nodes, the edge resource capacities are prone to be under-utilized for two reasons.
First, the heterogeneity of edge resource capacities can cause resource under-utilization.
For an edge node that has insufficient storage capacity to cache a service or cannot provide sufficient computation capacity for an application, the corresponding computation tasks have to be outsourced to central clouds rather than to nearby powerful edge nodes, resulting in under-utilization of edge resources \cite{xu2018joint}.
Moreover, the inconsistence of storage and computation capacities of edge nodes further aggravates edge resource wasting.
An edge node with large computation capacity cannot process substantial computation tasks when it has insufficient storage capacity to cache the services, leading to under-utilization of edge computation capacities.
To fully utilize both the storage and computation capacities of edge nodes, it is crucial to explore the potential of cooperation among edge nodes .
\begin{figure}
  \centering
  \includegraphics[width=0.6\textwidth]{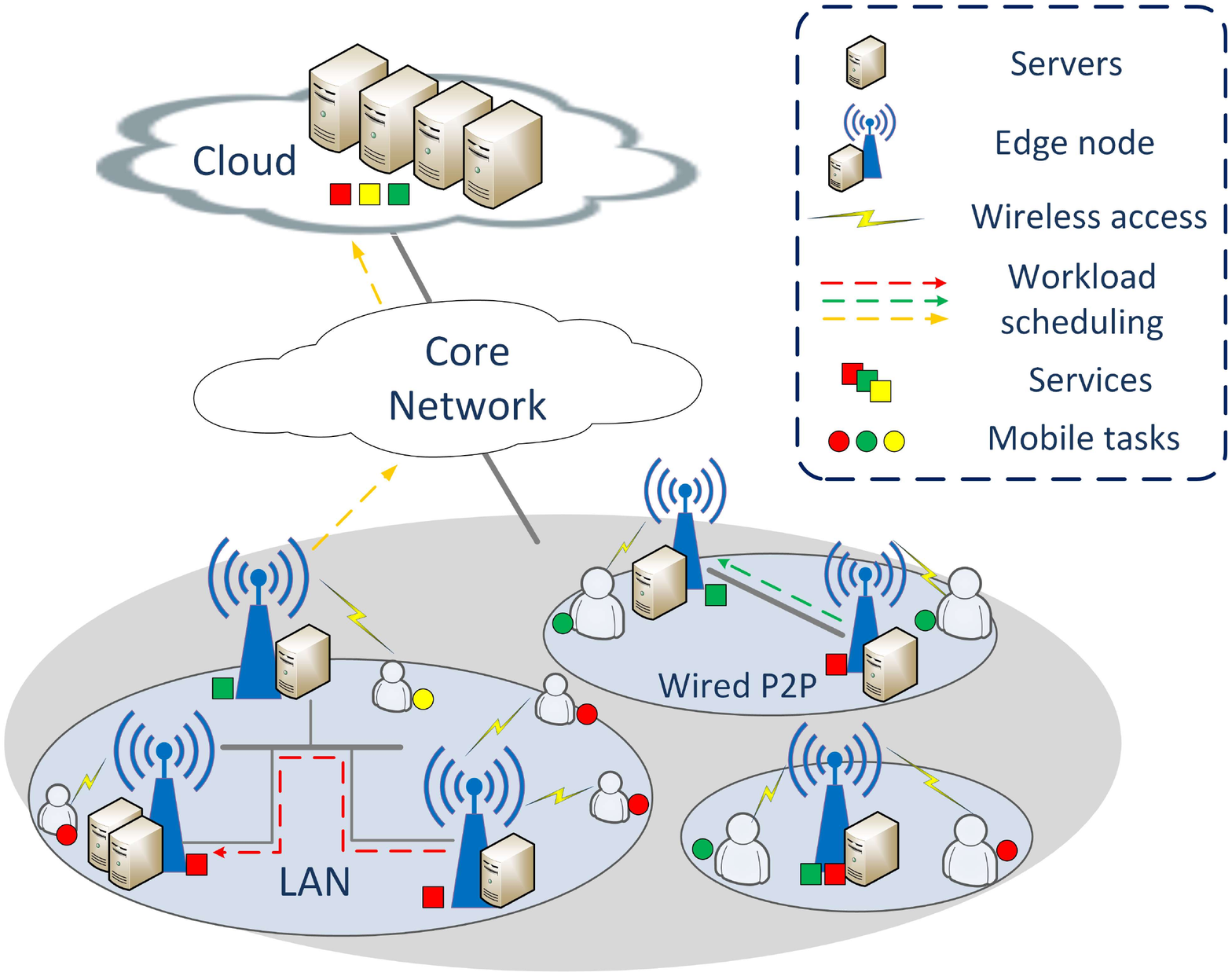}\\
  \caption{Cooperative service caching and workload scheduling in mobile edge computing.}\label{fig:overview}
\end{figure}

In this paper, we consider cooperation among edge nodes and investigate cooperative service caching and workload scheduling in mobile edge computing.
As shown in Fig. \ref{fig:overview}, nearby edge nodes are connected by local area network or wired peer-to-peer connection \cite{chen2018computation}.
For an edge node that is not caching a service or does not provide sufficient computation capacity, the corresponding computation tasks can be offloaded to nearby under-utilized edge nodes that have cached the service or outsourced to the cloud.
Through exploiting the cooperation among edge nodes, the heterogeneous edge resource capacities can be fully utilized and the resource capacity inconsistence of individual edge nodes can be alleviated.
The existing work which considers edge cooperation and jointly optimizes service caching and workload scheduling has sought to maximize the overall requests served at edge nodes while ensuring the service caching cost within the budget \cite{he2018s,farhadi2019service}.
However, it is hard to determine the exact value of the budget in practical scenarios.
Furthermore, while the reduced delay is the main advantage of mobile edge computing, the service response time is not considered as a performance criteria in the existing work.
In this paper, we investigate the cooperative service caching and workload scheduling with the objective of minimizing the service response time as well as the outsourcing traffic (denoted as problem 1).

Solving this problem is challenging in three folds.
First, service caching and workload scheduling are coupled.
Service caching policies determine the decision space of workload scheduling,
and in return, the workload scheduling results reflect the performance of the service caching policies.
Solving problem 1 needs to consider the interplay between the two subproblems.
Second, minimizing the service response time requires to properly trade off the computation and the transmission delay.
While offloading computation tasks from overloaded edge nodes to nearby under-utilized edge nodes is beneficial reduce the computation delay, task offloading causes additional transmission delay on LAN.
Solving problem 1 optimally should deal with the computation-communication tradeoff.
Third, solving problem 1 needs to deal with the heterogeneity of edge resource capacities.
Edge nodes are heterogeneous in both the storage and computation capacities.
Optimizing problem 1 needs to balance the workloads among the heterogeneous edge nodes, causing exponential computation complexity.
How to deal with edge heterogeneity and design algorithms with reduced computation complexity is challenging.

To deal with the challenge of subproblem coupling, we formulate problem 1 as a mixed integer nonlinear programming problem to jointly optimize service caching and workload scheduling.
A two-layer Iterative Caching updatE (ICE) algorithm is designed to illustrate the interplay of the two subproblems, with the outer layer updating the edge caching policies iteratively based on Gibbs sampling (the service caching subproblem) and the inner layer optimizing the workload scheduling polices (the workload scheduling subproblem).
To properly trade off the computation and communication delay, we use queuing models to analyze the delay in each part of the system and thereby compute the average service response time.
A proper computation-communication tradeoff can be achieved when the average service response time is minimized.
To deal with the exponential complexity of workload scheduling caused by edge heterogeneity, we exploit the convexity of the workload scheduling subproblem and propose a heuristic workload scheduling algorithm with polynomial computation complexity based on the idea of water filling.

The contributions of this paper are summarized as follows:
\begin{itemize}
  \item We investigate cooperative service caching and workload scheduling in mobile edge computing, aiming at minimizing service response time and outsourcing traffic.
      We formulate this problem as a mixed integer non-linear programming problem and show the non-polynomial complexity by analyzing simplified cases of this problem.
  \item We use queuing models to analyze the delay in each part of the system, based on which the convexity of the workload scheduling subproblem is proved.
  \item We propose two-layer ICE to solve problem 1, with the outer layer updating the service caching policies iteratively based on Gibbs sampling and inner layer optimizing the workload scheduling policies.
      By exploiting the convexity of the workload scheduling subproblem, we further propose a heuristic workload scheduling algorithm with reduced computation complexity based on the idea of water-filling.
  \item We conduct extensive simulations to evaluate the effectiveness and the convergence of the proposed algorithm.
\end{itemize}

This paper is organized as follows.
Section II reviews the related work.
Section III analyzes the system model and provides problem formulation.
In Section IV, algorithm design is presented in detail, and Section V illustrates simulation results.
Finally, the concluding remarks are given in Section VI.
\section{Related Work}
Mobile edge computing has been envisioned as a promising computing paradigm with the benefits of reduced delay and lower outsourcing traffic.
Due to the limited storage and computation capacities of edge nodes, properly placing services of mobile applications and scheduling computation tasks among edge nodes are crucial to optimize the quality of services with high resource efficiency.
There has been extensive work devoted to workload scheduling, service caching, or joint service caching and workload scheduling.

Although mobile edge computing enables mobile users to access powerful resources within one-hop range \cite{liang2017}, a lot of prior work has evolved to allow task offloading to edge nodes (or the remote cloud) within more than one hop and solve the workload scheduling problem.
Online workload scheduling among edge-clouds has been studied in \cite{wang2017online,tan2017online} to accommodate dynamic requests in mobile edge computing
In \cite{tong2016hierarchical}, Tong \emph{et al.} have developed a hierarchical architecture of edge cloud servers and optimized workload placement in this architecture.
Cui \emph{et al.} \cite{cui2017software} have proposed the software defined control over request scheduling among the cooperative mobile cloudlets.
All the above work on workload scheduling has a common assumption that each edge node (also named as edge-cloud or cloudlet) have cached all the services and can process any types of computation tasks, which is impractical due to the limited storage capacities of edge nodes.
Service caching among edge nodes should also be taken into consideration.

Caching services at edge nodes is an effective approach to relieving the burden of backhual network and the central clouds, and increasing efforts are devoted to edge service caching.
Borst \emph{et al.} \cite{borst2010distributed} have presented popularity-based distributed caching algorithms for content distribution network (CDN).
Dynamic edge service caching has been extensively studied in \cite{dan2014dynamic,hou2016asymptotically,wang2016dynamic,zhang2013dynamic}.
Prediction-based content placement has been investigated in \cite{dan2014dynamic} and approximations of dynamic content allocation have been provided for the hybrid system of cloud-based storage and CDN.
History-based dynamic edge caching have been proposed in \cite{hou2016asymptotically} without predicting future requests or adopting stochastic models.
In mobile edge computing, due to the limitation of both storage and computation capacities of edge nodes, service caching and workload scheduling should be jointly optimized to improve the system quality of service with high resource efficiency.

Joint optimization of edge caching and request routing for data-intensive applications (such as video streaming) has been studied in \cite{dehghan2017complexity} to minimize the average access delay.
Nevertheless, this work cannot be directly applied to applications which are both data-intensive and computation-intensive (such as augmented reality) and need to consider computation delay at edge nodes.
To address the above issue, joint optimization of service caching and workload scheduling have been investigated in \cite{xu2018joint,he2018s,farhadi2019service}.
 The work \cite{xu2018joint} has jointly optimized service caching and task offloading without considering cooperation among edge nodes, which can lead to under-utilization of heterogeneous edge resource capacities.
The work \cite{he2018s} and \cite{farhadi2019service} have investigated joint service caching and request scheduling without taking the service response time (including the transmission delay and the computation delay) as the performance criteria, which cannot highlight the benefit of reduced delay in mobile edge computing.
Different from the existing work, we study the cooperative service caching and workload scheduling in mobile edge computing, aiming at minimizing the service response time and the outsourcing traffic to central clouds.
We solve this problem by developing the iterative caching update algorithm based on Gibbs sampling and further proposing the heuristic workload scheduling algorithm with polynomial complexity based on the idea of water filling.
%
%

\section{System Model and Problem Formulation}
\subsection{System Model}
In this paper, we investigate cooperative service caching and workload scheduling in mobile edge computing.
As shown in Fig. \ref{fig:overview}, nearby edge nodes are connected by local area network or wired peer-to-peer connection. For an edge node that is not caching a service or does not provide sufficient computation capacity, the corresponding computation tasks can be offloaded to nearby under-utilized edge nodes that have cached the
service or outsourced to the cloud.
We consider a multi-edge system consisting of a set of $\mathds{N}=\{1,2,...,N\}$ edge nodes, each of which is equipped with the  computation capacity  $R_n$ ($n\in \mathds{N}$) and storage capacity $P_n$ ($n\in \mathds{N}$).
The system provides a library of $\mathds{S}=\{1,2,...S\}$ services, such as mobile gaming, object recognition, video streaming, etc, which are differentiated by the computation and storage requirements.
To process a type of mobile application at network edge, an edge node should provision certain storage capacity to cache the service of the application.
Let $p_s$ be the required storage capacity to cache service $s$.
For each service $s$, we consider that the computation requests of the corresponding computation task (in CPU cycles) follow exponential distribution with the expectation of $\beta_s$, and the task arrival at each edge node $n$ is a Poisson process with the expected rate $A_{ns}$, which is a general assumption \cite{xu2018joint}.
There is a centralized cloud with ample storage and computation capacity, thus the cloud stores all the services and the processing delay in the cloud $d_{\rm{cloud}}$ is mainly caused by the transmission delay from edge nodes to the cloud.
\subsubsection{Edge Caching and Workload Scheduling Policies}
Two questions should be answered in this study: 1) which edge nodes cache each type of service? and 2) how to schedule the computation workloads among the connected edge nodes that have cached the same services?
We use two set of variables to model the edge caching and workload scheduling results: $c_{ns}$ indicates whether service $s$ is cached at edge node $n$, and $\lambda_{ns}$ represents the workload ratio of service $s$ that are executed at edge node $n$.
We refer by \emph{edge caching} and \emph{workload scheduling} policies to the respective vectors:
\begin{equation}\label{equa:definition}
  \begin{array}{l}
\bm{C} = ({c_{ns}} \in \{ 0,1\} :n \in \mathds{N}, s \in \mathds{S}),\\
\bm{\Lambda }  = ({\lambda _{ns}} \in [ 0,1] :n \in \mathds{N}\cup\{o\}, s \in \mathds{S}).
\end{array}
\end{equation}

Denote by $\bm{c_n}=(c_{ns}: s\in \mathds{S})$ the caching decision of edge node $n$, and $\bm{C_n}$ the action space of $\bm{c_n}$, i.e., $\bm{c_n}\in \bm{C_n}$.
The services cached at each edge node cannot exceed the storage capacity, i.e.,
\begin{equation}\label{equa:storage constraint}
  \sum\limits_{s \in \mathds{S}} {{c_{ns}}{p_s}}  \le {P_n}.
\end{equation}
Let $\lambda_{os}$ denote the workload ratio of service $s$ outsourced to the cloud, there is
 \begin{equation}\label{equa:scheduling equation}
   \sum\limits_{n \in \mathds{N} \cup \{ o\} } {{\lambda _{ns}}}  =1.
 \end{equation}
\subsubsection{Service Response Time}
Denote by $\Theta_n$ the set of nearby edge nodes that have direct connection with edge node $n$, and $d_n$ the transmission delay on LAN to edge node $n$.
The computation workload executed at edge node $n$ should be no more than the overall arriving tasks of nearby edge nodes,
\begin{equation}\label{equa:nearby arrival constraint}
 {\lambda _{ns}}A_s \le \sum\limits_{i \in {\Theta _n} \cup \{ n\} } {{A_{is}}},
\end{equation}
where $A_s$ is the overall computation workload of service $s$ in the system, i.e. ${A_s} = \sum\limits_{n \in N} {{A_{ns}}} $.
We can notice that if $ {\lambda _{ns}}A_s\le A_{ns}$, all the tasks are from edge node $n$; Otherwise, the excessive tasks (${\lambda _{ns}}A_s- A_{ns}$) are from nearby edge nodes.

At each edge node, the computation capacity is shared by the cached services.
Let the function $\Gamma_n$ represent the computation allocation mechanism at edge node $n$, i.e., the computation capacity allocated to service $s$ is ${r _{ns}}=\Gamma_n(\bm{C})$.
 For each service $s$, as the computation requests of the responding computation task follow exponential distribution, the serving time at edge node $n$ also follows exponential distribution with the expectation $\frac{\beta_s}{r_{ns}}$.
 Moreover, the task arrival of service $s$ at edge node $n$ is a Poisson process with the expectation $\lambda_{ns}A_s$.
Thus for each service $s$, the serving process of computation tasks at edge node $n$ can be modeled as an M/M/1 queue, and the computation delay is
 \begin{equation}\label{equa: edge computation delay}
   D_{ns}=\frac{1}{\mu_{ns}-\lambda_{ns}A_s},
 \end{equation}
 where $\mu_{ns}=\frac{r_{ns}}{\beta_s}$.
 To ensure the stability of the queue, there should be
 \begin{equation}\label{equa: queue stability}
   \lambda_{ns}A_s<\mu_{ns}.
 \end{equation}
 By combining Eq. (\ref{equa:nearby arrival constraint}) and Eq. (\ref{equa: queue stability}), $\lambda_{ns}$ is constrained as
  \begin{equation}
   \lambda_{ns}A_s\le\min \{ \sum\limits_{i \in {\Theta _n} \cup \{ n\} } {{A_{is}}} ,{\mu _{ns}} - \varepsilon \},
 \end{equation}
 where $\varepsilon>0$.

When outsourcing tasks to the cloud, the processing time is mainly caused by the transmission delay in the core network.
Similar as the task arrivals at edge nodes, the task arrival in the core network is also a Poisson process, with the expected rate $\lambda_{os}A_s$.
Let $t_s$ be the amount of transmission requests (e.g. input data) when outsourcing one unit of computation requests for service $s$ (in CPU cycle).
Here, $t_s$ is a constant related to the specific service $s$ \cite{chen2018computation}, \cite{cui2017software}.
Then for the service $s$, the transmission requests of a corresponding task follow exponential distribution with the expectation $t_s\beta_s$.
The transmitting time of a task in the core network also follows exponential distribution with the expectation $\frac{t_s\beta_s}{B_s}$, where $B_s$ represents the core network bandwidth to transmit service $s$.
Hence, the transmitting delay in the core network is given as
\begin{equation}\label{equa:cloud delay}
  d_{\rm{cloud}}=\frac{1}{{\frac{{{B_s}}}{{{t_s}{\beta _s}}} - {\lambda _{os}}{A_s}}},
\end{equation}
where
\begin{equation}
{\lambda _{os}}{A_s} < \frac{{{B_s}}}{{{t_s}{\beta _s}}}.
\end{equation}

 The average response time of service $s$ can be computed as a weighted sum of delay at each part of the system, including the computation delay at edge nodes, the transmission delay on LAN and the transmission delay to the cloud, i.e.,
 \begin{equation}\label{equa: service response time}
   {D_s} = \sum\limits_{n \in \mathds{N}} {\left[ {{{\lambda _{ns}D_{ns}}} + \frac{{\max \{ {\lambda _{ns}} A_s- {A_{ns}},0\} }}{{{A_s}}}{d_n}} \right]}  + {{\lambda _{os}}}{d_{{\rm{cloud}}}}.
 \end{equation}
 Here, ${\frac{{\max \{ {\lambda _{ns}}A_s - {A_{ns}},0\} }}{{{A_s}}}}$ represents the ratio of the workload offloaded to edge node $n$ from nearby edge nodes.
\subsection{Problem Formulation}
 This paper jointly optimizes the edge service caching and workload scheduling policies, aiming at minimizing the service response time and the overall outsourcing traffic to the cloud:
 \begin{equation}\label{equa: problem formulation}
   \begin{aligned}
{\textbf{P1}}:&\mathop {\min }\limits_{{\bf{C}},{\bf{\Lambda }}} \sum\limits_{s \in \mathds{S}} {({D_s} + {w_s}{\lambda _{os}}A_s)} \\
s.t.&\sum\limits_{n \in N \cup \{ o\} } {{\lambda _{ns}}}  = 1,~~~~~~~~~~~~~~~~~~~~~~~s \in \mathds{S}\\
\rm{C1}:&~~\sum\limits_{s \in S} {{c_{ns}}{p_s}}  \le {P_n},~~~~~~~~~~~~~~~~~~~~~~n \in \mathds{N}\\
\rm{C2}:&~~ {\lambda _{ns}}A_s \le \min \{ \sum\limits_{i \in {\Theta _n} \cup \{ n\} } {{A_{is}}} ,{\mu _{ns}} - \varepsilon \} ,n \in \mathds{N},s \in \mathds{S}\\
\rm{C3}:&~~{\lambda _{os}}{A_s} \le \frac{{{B_s}}}{{{t_s}{\beta _s}}}- \varepsilon,\\
\rm{C4}:&~~ {\lambda _{ns}}\ge 0,~~~~~~~~~~~~~~~~~~~~~~~~~~~~~~n \in \mathds{N}\cup\{o\},s \in \mathds{S}\\
\rm{C5}:&~~{c_{ns}} \in \{ 0,1\}. ~~~~~~~~~~~~~~~~~~~~~~~~~~n \in \mathds{N},s \in \mathds{S}
\end{aligned}
 \end{equation}
Here $w_s$ is a weight constant which is positively related to the transmitted data traffic when outsourcing tasks of service $s$.
Constraint C1 ensures the cached services at each edge node do not exceed the storage capacity.
C2 is the combined result of Eq. (\ref{equa:nearby arrival constraint}) and Eq. (\ref{equa: queue stability}), ensuring that each edge node only admits computation requests from nearby edge nodes, and the computation workload scheduled to each edge node does not exceed the computation capacity for each service.
\subsection{Complexity Analysis}
Problem \textbf{P1} is a mixed integer nonlinear programming problem.
In this section, we present the non-polynomial computation complexity of \textbf{P1} by analyzing the simplified cases including \emph{non-cooperation among edge nodes} and \emph{considering one single type of service}.
\subsubsection{Simplified Case 1: Non-cooperation among Edge Nodes}
In the first case, we assume that there is no cooperation among edge nodes.
With this assumption, the computation tasks of different services are either processed locally or directly outsourced to the cloud.
Thus, the computation tasks outsourced to the cloud are not only decided by the edge computation capacity, but also highly dependent on the storage capacity of each individual edge node.
In this scenario, problem \textbf{P1} is reduced to the \emph{service caching and task oursourcing} problem, similar as \cite{xu2018joint}.
Specifically, workload scheduling among edge nodes in \textbf{P1} is reduced to $N$ independent task outsourcing subproblems.
Each edge node only needs to decide the oursourced computation requests $\lambda_{o_ns}$ ( which is given as $ \lambda_{o_ns}=1-\lambda_{ns}$) according to its own service caching policy and the computation capacity limitation.
It is indicated in \cite{xu2018joint} that the reduced \emph{service caching and task outsourcing} problem remains challenging since it is still a mixed integer nonlinear programming problem and has an non-polynomial computation complexity.
\subsubsection{Simplified Case 2: Considering One Single Type of Service}
In this simplified case, we assume that only one single type of service is considered in the system.
Then, the caching result at each edge node can be simply determined by the relationship of the service storage requirement and the edge storage capacity: The service is cached at one edge node if it has ample storage capacity; Otherwise, the service is not cached at the edge node.
With this assumption (i.e., the service caching policy \bm{$C$} is given), problem \textbf{P1} is reduced to a \emph{workload scheduling} problem, which schedules computation workloads among the edge nodes that have sufficient storage capacity to cache the service.

Solving the \emph{workload scheduling} problem is challenging in two aspects.
First, edge nodes are heterogeneous in both computation task arrivals and edge computation capacities.
Balancing the workloads among the heterogeneous edge nodes  is critical to minimize the service response time and the outsourced traffic to the cloud, which, however, can cause exponential computation complexity when achieved in a centralized manner.
Second, scheduling workloads among edge nodes should consider the computation-transmission tradeoff.
Offloading computation tasks from overloaded edge nodes to nearby light-loaded edge nodes or to the cloud is beneficial to reduce the computation delay, but meanwhile causes additional transmission delay.
Minimizing service response time demands to properly trade off the computation and transmission delay.

By summarizing the above two simplified cases of problem \textbf{P1}, both the reduced \emph{service caching and task outsourcing} and \emph{workload scheduling} problems have non-polynomial computation complexity.
Therefore, problem \textbf{P1} also has non-polynomial computation complexity and it is crucial to solve this problem with reduced computation complexity.
\section{Algorithm Design}
As clarified in the above section, even the simplified cases of problem \textbf{P1} remain to have non-polynomial computation complexity.
This section presents the main idea of algorithm design which jointly optimizes the service caching and workload scheduling policies with reduced computation complexity.
Specifically, we design a two-layer Iterative Caching updatE algorithm (ICE), with the outer layer updating service caching policies based on Gibbs sampling \cite{Lynch2007Introduction}.
In inner layer, the edge caching policies are given and problem \textbf{P1} is reduced to the workload scheduling subproblem among the edge nodes that have cached a certain type of service (similar to Simplified case 2).
We demonstrate the exponential computation complexity of the reduced problem with convexity analysis and further propose a heuristic workload scheduling algorithm (Algorithm \ref{algorithm:heuristic workload scheduling algorithm}) with reduced computation complexity based on the idea of water filling.
\subsection{Iterative Caching updatE Algorithm (ICE)}
Gibbs sampling is a Monte Carlo Markov Chain technique, which can deduce the joint distribution of several variables from the conditional distribution samples.
The main idea of Gibbs sampling is to simulate the conditional samples by sweeping through each variable while maintaining the rest variables unchanged in each iteration.
The Monte Carlo Markov Chain theory guarantees that the stationary distribution deduced from Gibbs sampling is the target joint distribution \cite{Gilks1996Markov}.
In this work, we exploit the idea of Gibbs sampling to determine the optimal service caching policies iteratively, as shown in Algorithm \ref{algorithm:iterative caching update algorithm}.
The key point of the algorithm is to associate the conditional probability distribution of edge caching policies with the objective of \textbf{P1} (Step 7).
Through properly designing the conditional probability in each iteration, the deduced stationary joint distribution can converge to the optimal edge caching policies with high probability.

The ICE algorithm works as follows.
In each iteration, randomly select an edge node $n$ ($n\in \mathds{N}$) and a feasible edge caching decision $\bm{c_n}^*$ while maintaining the caching decisions of the rest edge nodes unchanged (Step 3).
With the given caching policies of all the edge nodes, \textbf{P1} is reduced to the workload scheduling subproblem:
\begin{equation}\label{equa: workload scheduling}
   \begin{aligned}
{\textbf{P2}}:&\mathop {\min }\limits_{{\bf{\Lambda }}} \sum\limits_{s \in \mathds{S}} {({D_s} + {w_s}{\lambda _{os}}A_s)} \\
s.t.&\sum\limits_{n \in N \cup \{ o\} } {{\lambda _{ns}}}  = 1,~~~~~~~~~~~~~~~~~~~~~~~~~~~s \in \mathds{S}\\
&~~ {\lambda _{ns}}A_s \le \min \{ \sum\limits_{i \in {\Theta _n} \cup \{ n\} } {{A_{is}}} ,{\mu _{ns}} - \varepsilon \} ,n \in \mathds{N},s \in \mathds{S}\\
&~~{\lambda _{os}}{A_s} \le \frac{{{B_s}}}{{{t_s}{\beta _s}}}- \varepsilon,\\
&~~ {\lambda _{ns}}\ge 0.~~~~~~~~~~~~~~~~~~~~~~~~~~~~~~n \in \mathds{N}\cup \{o\},s \in \mathds{S}\\
\end{aligned}
 \end{equation}
 After solving \textbf{P2}, we can compute the optimal objective value $y$ (defined as $y=\mathop{\min}\limits_{{\bf{\Lambda }}} \sum\limits_{s \in \mathds{S}} {({D_s} + {w_s}{\lambda _{os}})}$).
Assume that when the selected edge node $n$ changes its caching decision from $\bm{c_n}$ to $\bm{c_n}^*$, the optimal objective value varies from $y$ to $y^*$.
Associate the conditional probability distribution of edge caching policies with the objective value as:
the selected edge node $n$ changes its caching decision from $\bm{c_n}$ to $\bm{c_n}^*$ with the probability $\rho=\frac{1}{{1{\rm{ + }}{e^{({{y}^{*}} - y)/\omega }}}}$ ($\omega>0$) and maintains the current caching decision $\bm{c_n}$ with $1-\rho$ (Step 7).
Finally, the iteration is ended if the stop criteria is satisfied.
%

\begin{algorithm}
\caption{Iterative Caching updatE Algorithm (ICE)}
\label{algorithm:iterative caching update algorithm}
\begin{algorithmic}[1]
\REQUIRE ~~ \\
$A_{ns}~(n\in\mathds{N}, s \in \mathds{S})$, $p_s$, $\beta_s$ ($s \in \mathds{S})$)\\
\ENSURE ~~ \\
The edge caching policy \bm{$C$} and the workload scheduling policy \bm{$\Lambda$}.
\STATE Initialize $\bm{C}^0\leftarrow \bm{0}$.\\
\FOR {iteration $i=1,2,...$}
\STATE Randomly select an edge node $n\in\mathds{N}$ and an edge caching decision $\bm{ {{c}}}_n^{*}\in \bm{C}_n$.
\IF {$\bm{ {{c}}}_n^{*}$ is feasible}
\STATE Based on the edge caching policy $(\bm{c}_1^{i-1},..,\bm{c}_n^{i-1},..,\bm{c}_N^{i-1})$, compute the optimal workload scheduling policy \bm{$\Lambda$} and the responding $y$ by solving \textbf{P2}.
\STATE Based on the edge caching policy $(\bm{c}_1^{i-1},...,\bm{ {{c}}}_n^{*},...,\bm{c}_N^{i-1})$, compute the optimal workload scheduling policy $\bm{{\Lambda}}^{*}$ and the responding ${y^{*}}$ by solving \textbf{P2}.
\STATE Let $\bm{c}_n^{i}=\bm{ {{c}}}_n^{*}$ with the probability $\rho=\frac{1}{{1{\rm{ + }}{e^{({{y}^{*}} - y)/\omega }}}}$, and $\bm{c}_n^{i}=\bm{c}_n^{i-1}$ with the probability $1-\rho$.
\ENDIF
\IF {the stopping criteria is satisfied}
\STATE End the iteration and return $\bm{C}^i$, $\bm{\lambda}^i$.
\ENDIF
\ENDFOR
\end{algorithmic}
\end{algorithm}
ICE has the following property.
\begin{Theo}\label{theorem:convengence property}
 ICE can converge to the globally optimal solution of problem \textbf{P1} with a higher probability as $\omega$ decreases.
When $\omega\rightarrow 0$, the algorithm converges to the globally optimal solution with the probability of 1.
\end{Theo}
\begin{proof}
Please refer to Appendix \ref{proof:proof of convergence property}.
\end{proof}
\noindent
Remark: Theorem \ref{theorem:convengence property} demonstrates that in each iteration of the Gibbs sampling technique, through properly selecting $\omega$ in $\rho=\frac{1}{{1{\rm{ + }}{e^{({y^*} - y)/\omega }}}} ~(\omega>0)$ which associates the service caching update process with the objective value, the algorithm can converge to the optimal edge caching policy with high probability.
\subsection{Heuristic Workload Scheduling Algorithm}
When the edge caching policy is given, problem \textbf{P2} should be solved to compute the optimal workload scheduling policy and the corresponding object value.
In this part, we first demonstrate the exponential complexity of \textbf{P2} through theoretical analysis and further propose a heuristic workload scheduling algorithm by exploiting the convexity of the problem.
\subsubsection{Computation Complexity of \textbf{P2}}
Substitute Eq. (\ref{equa: edge computation delay}) and (\ref{equa: service response time}) into \textbf{P2}, and the objective function $f(\bm{\Lambda})$ can be rewritten as
\begin{equation}\label{equa: objective function}
\begin{aligned}
f(\bm{\Lambda})=&\sum\limits_{s \in S} {({D_s} + {w_s}{\lambda _{os}}{A_s})}  \\
=& \sum\limits_{s \in S} {\sum\limits_{n \in N} {(\frac{{{\lambda _{ns}}}}{{{\mu _{ns}} - {\lambda _{ns}}{A_s}}} + \frac{{\max \{ {\lambda _{ns}}{A_s} - {A_{ns}},0\} }}{{{A_s}}}{d_n}} )}
\\&+ \sum\limits_{s \in S} {({\frac{{{\lambda _{os}}}}{{\frac{{{B_s}}}{{{t_s}{\beta _s}}} - {\lambda _{os}}{A_s}}}} + {w_s}{\lambda _{os}}{A_s})}.
\end{aligned}
\end{equation}
\begin{Theo}
\label{theorem:convexity property}
Problem \textbf{P2} is a convex optimization problem over the workload scheduling policy \bm{$\Lambda$}.
\end{Theo}
\begin{proof}
Please refer to Appendix \ref{proof:proof of convexity}.
\end{proof}
A convex optimization problem can be solved by searching for results satisfying the Karush-Kuhn-Tucker (KKT) conditions \cite{boyd2004convex}.
We first provide the KKT conditions of \textbf{P2}.
When the caching policy is given, the computation resources allocated to each service are determined according to $\Gamma_n(\bm{C})$.
Thus for one service, the workload scheduling policy among edge nodes that have cached the service is independent of the other services.
Solving problem \textbf{P2} is equivalent to optimizing the workload scheduling policy for each type of service.
Task a service $s$ ($s\in\mathds{S}$) as the representative.
Define the Lagrange function as
\begin{equation}\label{equa:Lagrange function}
\begin{aligned}
&{L_s}(\bm{\lambda _s},\bm{\alpha _s},\bm{\eta _s}) \\
=& ({D_s} + {w_s}{\lambda _{os}}) + \sum\limits_{n \in \mathds{N} \cup \{ o\} } {{\alpha _{ns}}({\lambda _{ns}}{A_s} - {\pi _{ns}})} \\
 &- \sum\limits_{n \in \mathds{N} \cup \{ o\} } {{\alpha _{(N + 1 + n)s}}{\lambda _{ns}}}  + {\eta _s}(\sum\limits_{n \in \mathds{N} \cup \{ o\} } {{\lambda _{ns}}}  - 1),
  \end{aligned}
\end{equation}
where $\bm{\alpha _s}$ and $\bm{\eta _s}$ are Lagrange multipliers, and ${\pi _{ns}}$ is the upper bound of the inequation constraints defined as ${\pi _{ns}}=\min \{ \sum\limits_{i \in {\Theta _n} \cup \{ n\} } {{A_{is}}} ,{\mu _{ns}} - \varepsilon \}$ ($n\in\mathds{N}$) and ${\pi _{os}}=\frac{{{B_s}}}{{{t_s}{\beta _s}}}- \varepsilon$.\\
Then the KKT conditions are given as
\begin{equation}\label{equa:KKT conditions}
\begin{aligned}
(\rm{C1})~&\frac{{\partial {L_s}(\bm{\lambda _s},\bm{\alpha _s},\bm{\eta _s})}}{{\partial {\lambda _{ns}}}} = 0,~~~~~~~n \in \mathds{N} \cup \{ o\} \\
(\rm{C2})~&0 \le {\lambda _{ns}}{A_s} \le {\pi _{ns}},~~~~~~~~~~~n \in \mathds{N} \cup \{ o\}\\
(\rm{C3})~&\sum\limits_{n \in N \cup \{ o\} } {{\lambda _{ns}}}   = 1,~~~\\
(\rm{C4})~&{\alpha _{ns}}({\lambda _{ns}}{A_s} - {\pi _{ns}}) = 0,~~~~~n \in \mathds{N} \cup \{ o\} \\
(\rm{C5})~&{\alpha _{(N + 1 + n)s}}{\lambda _{ns}} = 0,~~~~~~~~~~n \in \mathds{N} \cup \{ o\} \\
(\rm{C6})~&{\alpha _{ns}} \ge 0.~~~~~~~~~~~~~~~~~~~~~~~n = \{ 1,2...,2N + 2\}
\end{aligned}
\end{equation}
Here, (C4), (C5) and (C6) arise from the inequation constraints of \textbf{P2}.
For each inequation constraint, there are two possible results in Eq. (\ref{equa:KKT conditions}): 1) ${\alpha _{ns}} = 0,~{\lambda _{ns}}{A_s} < {\pi _{ns}}$ (or $ {\lambda _{ns}}>0$), indicating that the optimal results are at the extreme points derived from (C1); 2) ${\alpha _{ns}} > 0$,${\lambda _{ns}}{A_s} = {\pi _{ns}}$ (or $ {\lambda _{ns}}=0$), indicating the optimal results are at the boundary.
As there are $2(N+1)$ inequation constraints in problem \textbf{P2} (i.e., the computation capacity constraints of edge nodes), directly searching for the results satisfying the KKT conditions can cause $O(2^{2(N+1)})$ computation complexity .
To reduce the computation complexity of \textbf{P2}, we propose the heuristic workload scheduling algorithm.
\subsubsection{Algorithm Design}
The main idea of the algorithm is to first remove the computation capacity constraints of edge nodes and the transmission bandwidth constraint of the core network  (i.e., the inequation constraints in \textbf{P2}) to derive the correlation of workload scheduling results of edge nodes and the cloud.
Then we search for the optimal results satisfying the KKT conditions within the resource constraints.

When removing the inequation constraints, the KKT conditions only keep (C1) and (C3) in (\ref{equa:KKT conditions}), with (C1) changed to
\begin{equation}\label{equa:changed (C1)}
  \frac{{\partial {L_s}(\bm{\lambda _s},\bm{\eta _s})}}{{\partial {\lambda _{ns}}}} =
  \frac{{\partial ({D_s} + {w_s}{\lambda _{os}} + {\eta _s}(\sum\limits_{n \in \mathds{N} \cup \{ o\} } {{\lambda _{ns}}}  - 1))}}{{\partial {\lambda _{ns}}}}{\rm{ = 0}},
\end{equation}
for each $n \in \mathds{N} \cup \{ o\}$.
However, ${L_s}(\bm{\lambda _s},\bm{\eta _s})$ is not partially derivable over $\lambda_{ns}$ when $\lambda_{ns}=\frac{A_{ns}}{A_s}$ ($n\in \mathds{N}$), which is caused by $\frac{{\max \{ {\lambda _{ns}} A_s- {A_{ns}},0\} }}{{{A_s}}}$ in (\ref{equa: service response time}).
We solve this problem by dividing into two cases: $\lambda_{ns}<\frac{A_{ns}}{A_s}$ and $\lambda_{ns}\ge\frac{A_{ns}}{A_s}$, and $\lambda_{ns}({\eta _s})$ ($n\in \mathds{N}$) can be derived as
\begin{equation}\label{equa:piecewise lambda_n}
  \lambda {\;_{ns}} = \left\{ \begin{aligned}
\frac{1}{{{A_s}}}({\mu _{ns}} - \sqrt { - \frac{{{\mu _{ns}}}}{{{\eta _s}+{d_n}}}} )&{\rm{      if }}~{\eta _s} \le  - \frac{{{\mu _{ns}}}}{{{{({\mu _{ns}} - {A_{ns}})}^2}}} - {d_n}\\
\frac{1}{{{A_s}}}({\mu _{ns}} - \sqrt { - \frac{{{\mu _{ns}}}}{{{\eta _s}}}} )~~~~~&{\rm{      if }}~{\eta _s} \ge  - \frac{{{\mu _{ns}}}}{{{{({\mu _{ns}} - {A_{ns}})}^2}}}\\
\frac{{{A_{ns}}}}{{{A_s}}},~~~~~~~~~~~~~~~~~~~~~~~&{\rm{otherwise}}
\end{aligned} \right.
\end{equation}
and $\lambda_{os}$ is given as
\begin{equation}\label{equa:piecewise lambda_o}
 \lambda _{{os}} = \frac{1}{{{A_s}}}(\frac{{{B_s}}}{{{t_s}{\beta _s}}} - \sqrt { - \frac{{{B_s}}}{{{t_s}{\beta _s}({w_s} + {\eta _s})}}} ).
\end{equation}

\begin{figure}
  \centering
  \includegraphics[width=0.5\textwidth]{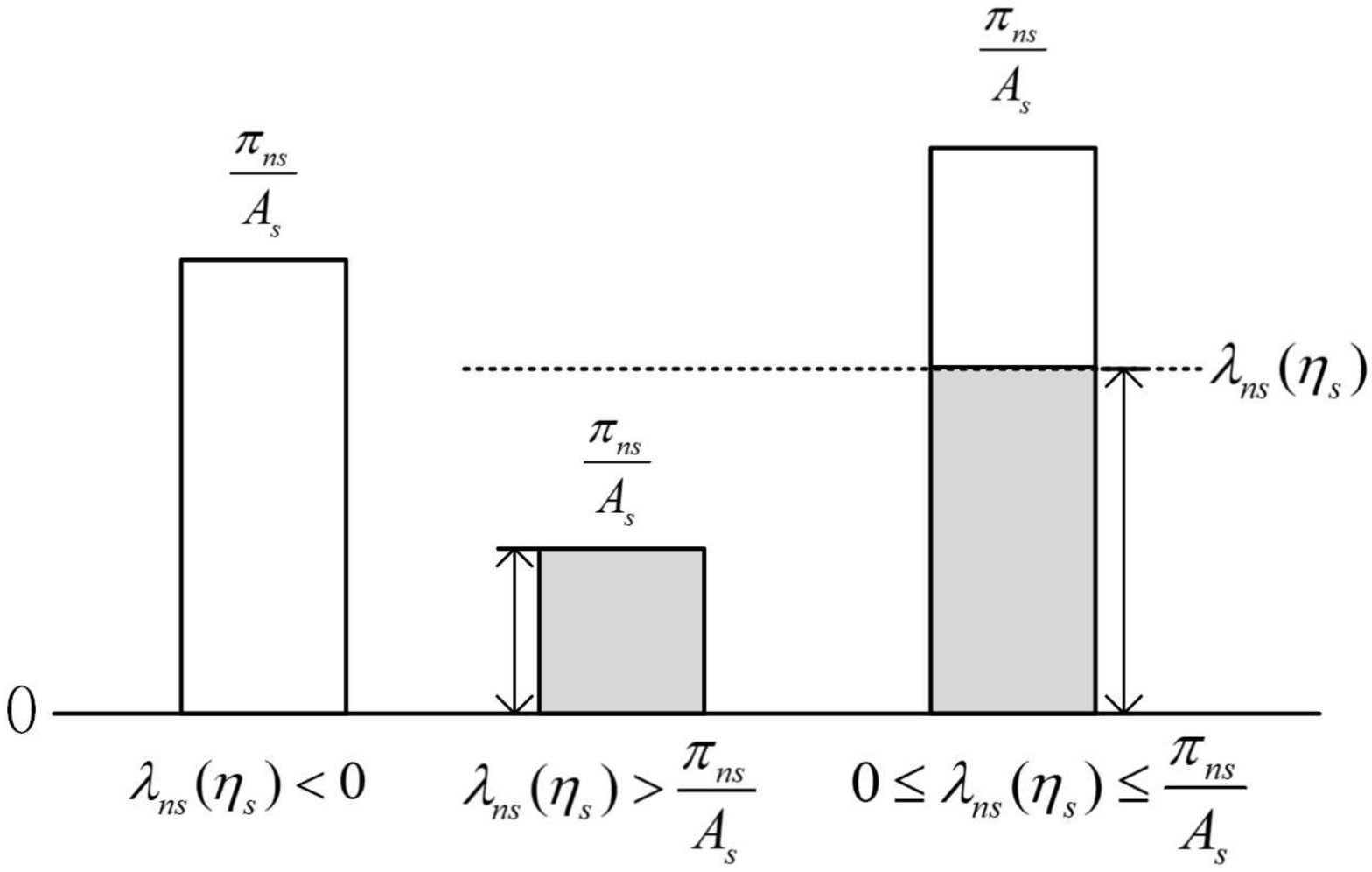}\\
  \caption{Water-filling based workload scheduling.}\label{fig:water-filling}
\end{figure}

After removing the inequation constraints, the workload scheduling policy $\bm{\Lambda}$ is given as the functions of $\eta_s$ ( Eq. (\ref{equa:piecewise lambda_n}), (\ref{equa:piecewise lambda_o})) to satisfy (C1) in the KKT conditions.
To obtain the optimal solution of $\eta_s$ which satisfies the equation constraint and the inequation constraints in \textbf{P2}, we search for the workload scheduled to $n$ ($n\in\mathds{N}\cup\{o\}$) based on the idea of water filling.
As shown in Fig. \ref{fig:water-filling}, scheduling workloads to edge nodes (or the cloud) is similar to filling water to tubes.
When the water level is above the upper bound or beneath the lower bound of the tube, the water cannot be decreased or increased anymore.
By combining Eq. (\ref{equa:piecewise lambda_n}), (\ref{equa:piecewise lambda_o}) with (C2), and we have the following conclusion:
Let $y(\eta_s)=\sum\limits_{n \in N \cup \{ o\} } {{\lambda _{ns}(\eta_s)}}  - 1$, then $y(\eta_s)$ is constant or monotone decreasing with $\eta_s$ (The proof is omitted).
Thus, we can search the optimal $\eta_s$ by the bisection method with the details summarized in Algorithm \ref{algorithm:heuristic workload scheduling algorithm}.
\begin{algorithm}
\caption{Heuristic Workload Scheduling Algorithm}
\label{algorithm:heuristic workload scheduling algorithm}
\begin{algorithmic}[1]
\STATE Define $y(\eta_s)$ as follows.
\FOR{each $n\in\mathds{N}\cup \{o\}$}
\STATE Compute $\lambda_{ns}(\eta_s)$ according to Eq. (\ref{equa:piecewise lambda_n}) and Eq. (\ref{equa:piecewise lambda_o}).
\IF{$\lambda_{ns}(\eta_s)<0$}
\STATE $\lambda_{ns}=0$.
\ELSE
\IF{$\lambda_{ns}(\eta_s)>\frac{\pi_{ns}}{A_s}$}
\STATE $\lambda_{ns}=\frac{\pi_{ns}}{A_s}$.
\ELSE
\STATE $\lambda_{ns}=\lambda_{ns}(\eta_s)$.
\ENDIF
\ENDIF
\ENDFOR
\STATE $y(\eta_s)=\sum\limits_{n \in N \cup \{ o\} } {{\lambda _{ns}}}  - 1$.
\STATE Find $\eta_s^{\rm{l}}<\eta_s^{\rm{r}}<0$ satisfying $y(\eta_s^{\rm{l}})>0,~\eta_s^{\rm{r}}<0$.\\
\WHILE{$\eta_s^{\rm{r}}-\eta_s^{\rm{l}}\ge \xi$}
\STATE $\eta_s^{\rm{m}}=\frac{\eta_s^{\rm{l}}+\eta_s^{\rm{r}}}{2}$.
\IF{$y(\eta_s^{\rm{l}})\cdot y(\eta_s^{\rm{m}})<0$}
\STATE $\eta_s^{\rm{r}}=\eta_s^{\rm{m}}$.
\ELSE
\STATE $\eta_s^{\rm{l}}=\eta_s^{\rm{m}}$.
\ENDIF
\ENDWHILE
\STATE The optimal result of $\eta_s$ is $\eta_s^{*}=\frac{\eta_s^{\rm{l}}+\eta_s^{\rm{r}}}{2}$.
\STATE Repeat step 2-13 to compute the optimal workload scheduling policy $\bm{\Lambda^{*}}$.
\end{algorithmic}
\end{algorithm}
In Algorithm \ref{algorithm:heuristic workload scheduling algorithm}, $\lambda_{ns}$ ($n\in\mathds{N}\cup \{o\}$) is traversed to compute $y(\eta_s)$ (step 2-13).
From step 16 to 23, the optimal $\eta_s$ is computed in an iterative manner, with overall $O(\log(\frac{\eta_s^r-\eta_s^l}{\xi}))$ iterations.
Therefore, the computation complexity of Algorithm \ref{algorithm:heuristic workload scheduling algorithm} is $O(N\log(\frac{\eta_s^r-\eta_s^l}{\xi}))$


\section{Simulation Results}
In this section, extensive simulations are conducted to evaluate our algorithms.
We simulate a 100m $\times$ 100m area covered with 12 edge nodes which serves a total of 8 services.
The edge nodes are empowered by heterogeneous storage and computation capacities, both of which follow uniform distribution.
The total arrival rates of computation tasks at different edge nodes $A_n$ ($n\in\mathds{N}$) are uniformly distributed.
At each edge node, the popularity of services follow Zipf's distribution, i.e., $\chi_{ns}\propto r_s^{-\upsilon}$, where $r_s$ is the rank of service $s$ and $\upsilon$ is the skewness parameter \cite{farhadi2019service}.
Thus, the arrival rate of computation tasks of service $s$ at edge node $n$ can be computed as
$A_{ns}=\chi_{ns}\cdot A_n$, where $A_n$ is the total arrival rate of computation tasks at edge node $n$.
The main parameters are listed in Table \ref{table:parameters}.
\begin{table}
\caption{Simulation Parameters}
\label{table:parameters}
\renewcommand\arraystretch{1.2}
\centering
\begin{tabular}{l|c}
\hline
\textbf{Parameter}&\textbf{Value}\\
\hline
Service storage requirement, $p_s$ &[20, 80] GB\\
Service computation requirement, $\beta_s$ & [0.1, 0.5] Giga CPU cycles/task\\
Edge node storage capacity, $P_n$ & [100, 200] GB\\
Edge node computation capacity, $R_n$ &[50, 100] Giga CPU cycles\\
Data transmission ratio of service, $t_s$ &[0.1,1.0] Mb/GHz\\
Core network bandwidth for service, $B_s$ & 160 Mbps\\
Skewness parameter, $\upsilon$ & 0.5\\
Smooth parameter, $\omega$ & $10^{-6}$\\
\hline
\end{tabular}
\end{table}

We compare the performance of ICE with two benchmark algorithms.\\
\textbf{Non-cooperation algorithm} \cite{chen2018computation}: Edge nodes cache services according to Gibbs Sampling.
At each edge node, the computation workloads of a service are either processed locally or outsourced to the cloud.\\
\textbf{Greedy algorithm}: Edge nodes cache services according to popularity.
Popular services have higher priority to be cached at edge nodes.
For the cached services, each edge node optimizes the workloads processed locally and outsourced to the cloud to minimize the edge process delay and the outsourcing traffic.
\subsection{Performance Comparison}
We compare the three algorithms in terms of the objective value, total service response time and outsourcing traffic by varying the average arrival rate of tasks at edge nodes (i.e., average $A_n$), and the results are shown in Fig. \ref{fig: performance comparison}.

\begin{figure}
\centering
\subfigure[Objective value]
{ \label{fig:performance_y} 
\includegraphics[width=0.60\textwidth]{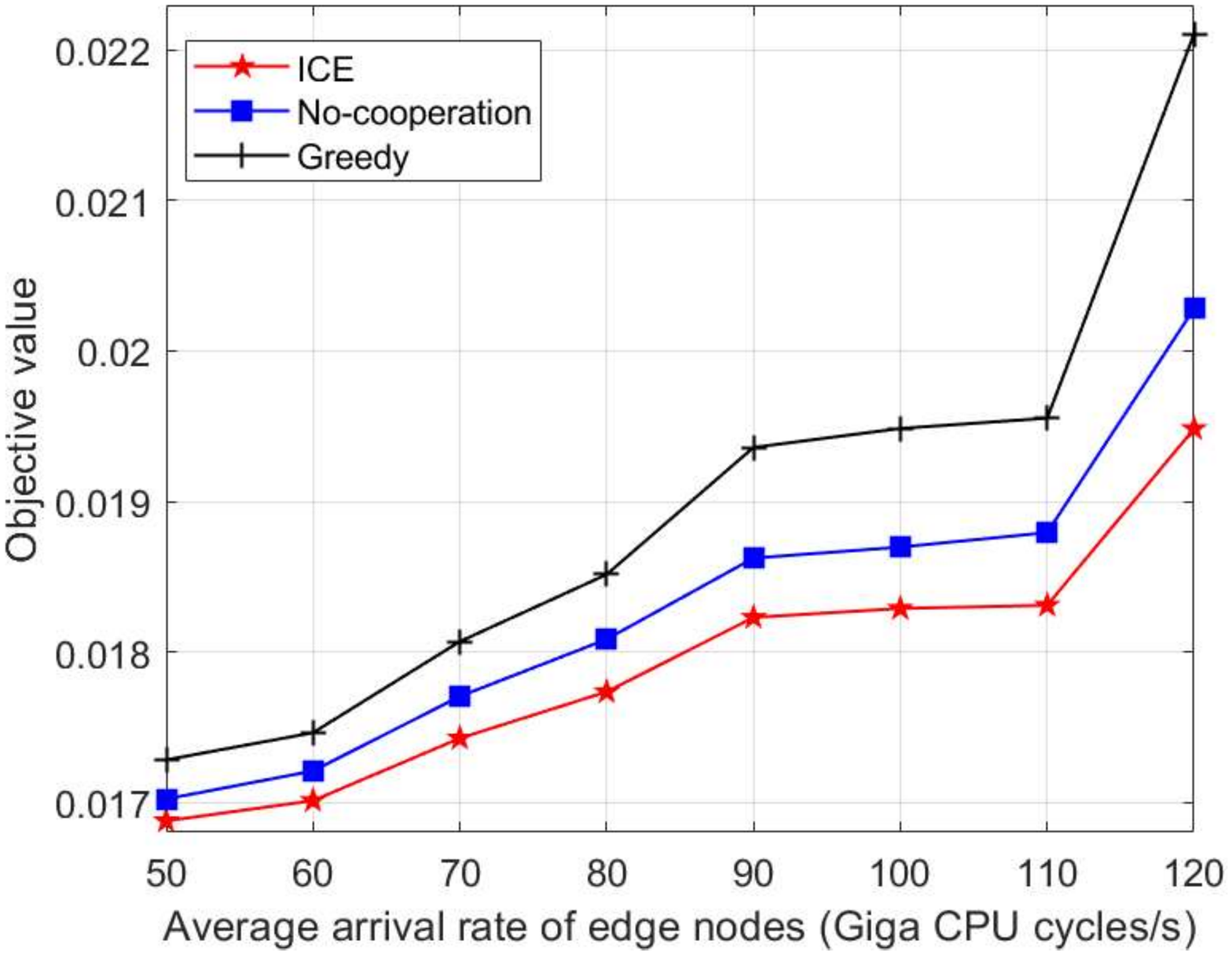}}
\hspace{0.01in}
\subfigure[Average response time]
{ \label{fig:performance_delay} 
\includegraphics[width=0.45\textwidth]{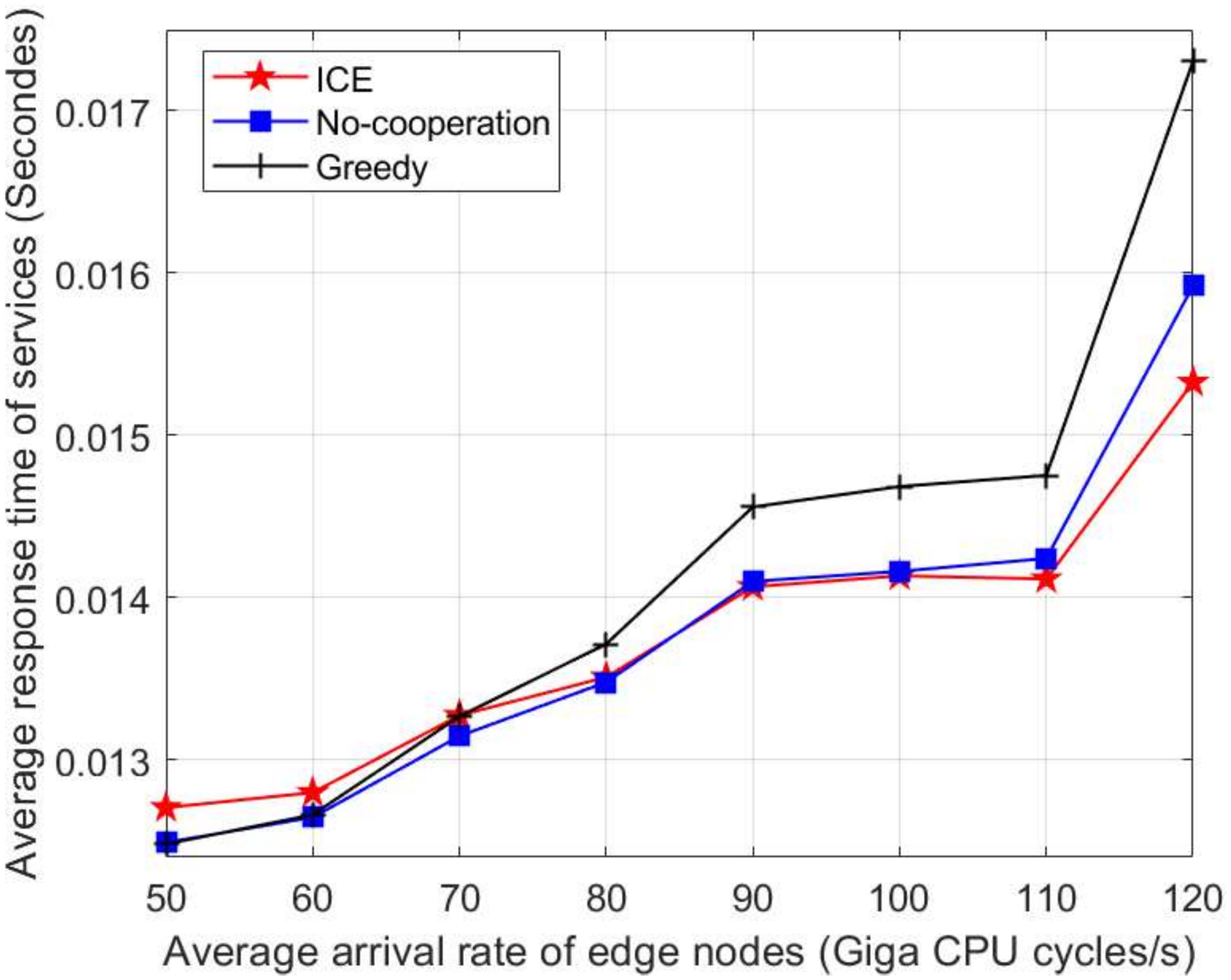}}
\subfigure[Average outsourcing traffic]
{ \label{fig:performance_out} 
\includegraphics[width=0.45\textwidth]{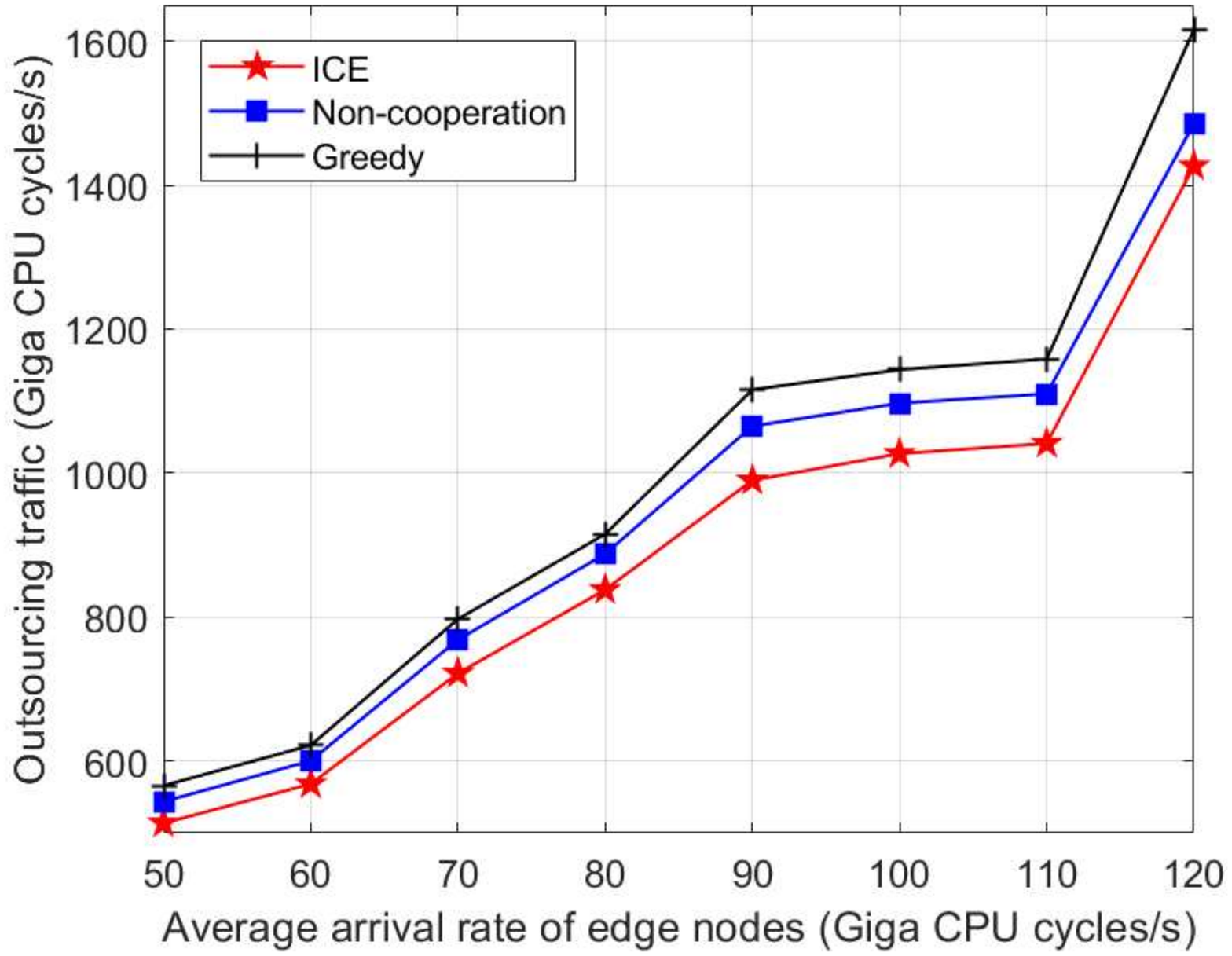}}
\caption{Performance comparison of different algorithms: the smooth parameter $\omega=10^{-6}$ and the weight factor $w_s=6\cdot10^{-4}$.} \label{fig: performance comparison} 
\end{figure}
Compared with the Non-cooperation algorithm and the Greedy algorithm, our Cooperation algorithm always yields the minimum object value and outsourcing traffic, and close to minimum total service response time .
In the Greedy algorithm, all the edge nodes cache the popular services with high priority, thus the computation tasks of less popular services have to be outsourced to the cloud.
Moreover, the Greedy algorithm only relies on service popularity to determine the edge caching policy without considering storage requirements of services.
Caching multiple less popular services with low storage requirements at edge nodes can be more beneficial to fully utilize both the computation and the storage capacities compared with caching one popular service with large storage requirement.
The Cooperation and Non-cooperation algorithms cache services based on Gibbs sampling, taking both the storage requirements of services and service popularity into consideration.
Therefore, the Greedy algorithm generally induces more outsourcing traffic and service response time than the other two algorithms.
The Non-cooperation algorithm cannot fully utilize the computation capacities of edge nodes which has low storage capacity due to the absence of cooperation among edge nodes.
In the Cooperation algorithm, both the storage and computation capacities of edge nodes can be coordinated and fully utilized through careful design of service caching and workload scheduling among the connected edge nodes.
\subsection{Convergence of ICE}
According to the theoretical analysis in Theorem \ref{theorem:convengence property}, the Gibbs sampling based service caching algorithm (Algorithm 1) can converge to the optimal service caching results with probability 1 when the smooth  parameter $\omega$ is close to 0.
This part illustrates the influence of $\omega$ on the convergence of ICE with the results shown in Fig. \ref{fig:convergence}.

As shown in Fig. \ref{fig:convergence}, the objective value can converge to the near-optimal results when $\omega\le10^{-4}$, and the converging rate is faster as $\omega$ decreases.
When $\omega \ge 10^{-3}$, the objective value converges slowly to higher value ($\omega=10^{-3}$) or even cannot converge ($\omega=10^{-2}$).
These results can be explained by Step 7 of ICE and Eq.(\ref{equa:rewritten}).
According to ICE, the smaller is $\omega$, the more probable that  the selected edge node updates to the better caching decision in each iteration.
Thus, when $\omega$ is small, the objective value converge quickly (within less iterations).
In addition, it can be concluded from Eq.(\ref{equa:rewritten}) that stationary probability of the optimal caching result increases with $\omega$, and the probability $\rightarrow 1$ when $\omega\rightarrow 0$.
Therefore, the smaller is $\omega$, the more probable that ICE converges to the optimal caching result.
\begin{figure}
  \centering
  \includegraphics[width=0.55\textwidth]{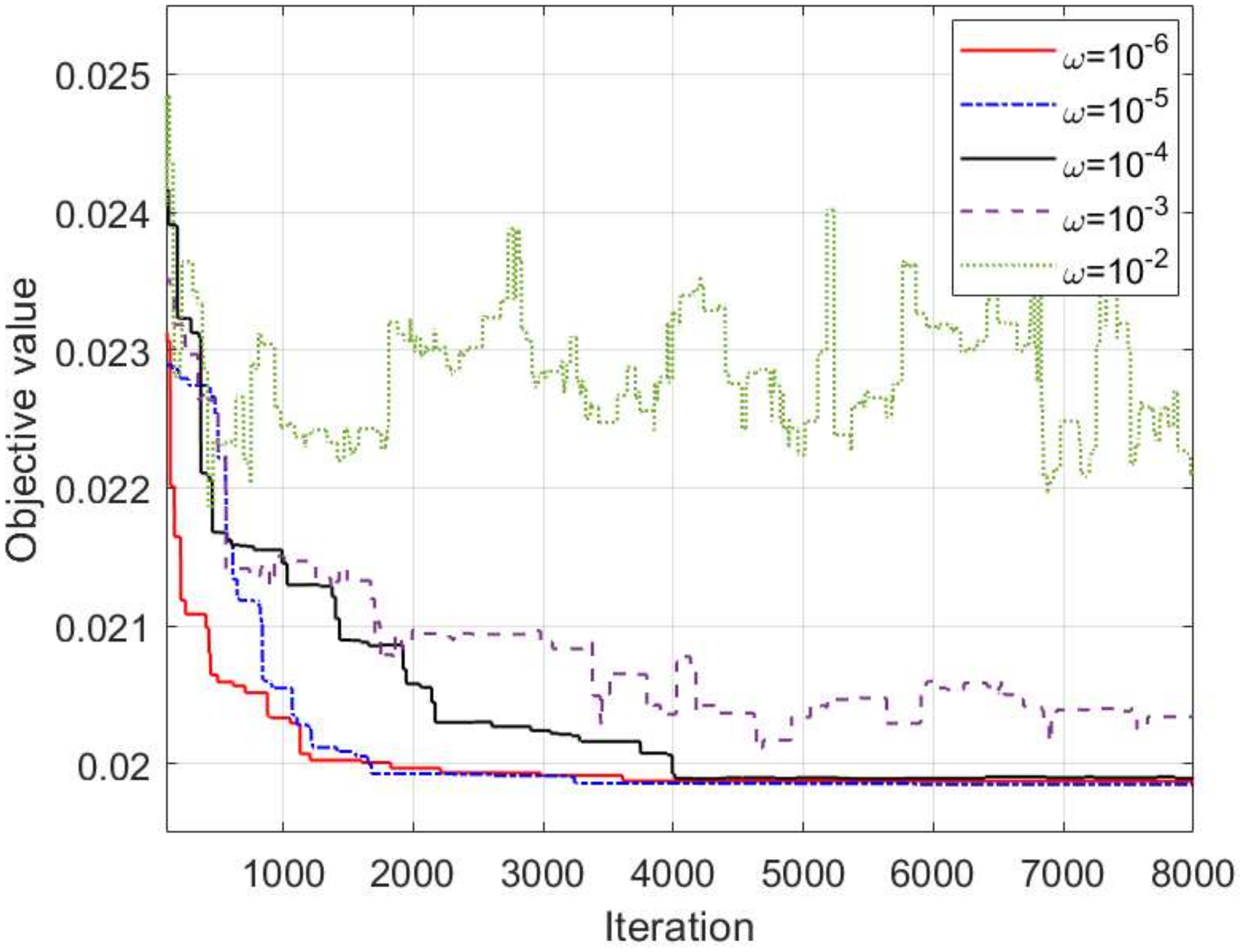}\\
  \caption{Impact of $\omega$ on the convergence of ICE: the weight factor $w_s=6\cdot 10^{-4}$.}\label{fig:convergence}
\end{figure}
\subsection{The Impact of Edge Node Connectivity}
\begin{figure}
  \centering
  \includegraphics[width=0.50\textwidth]{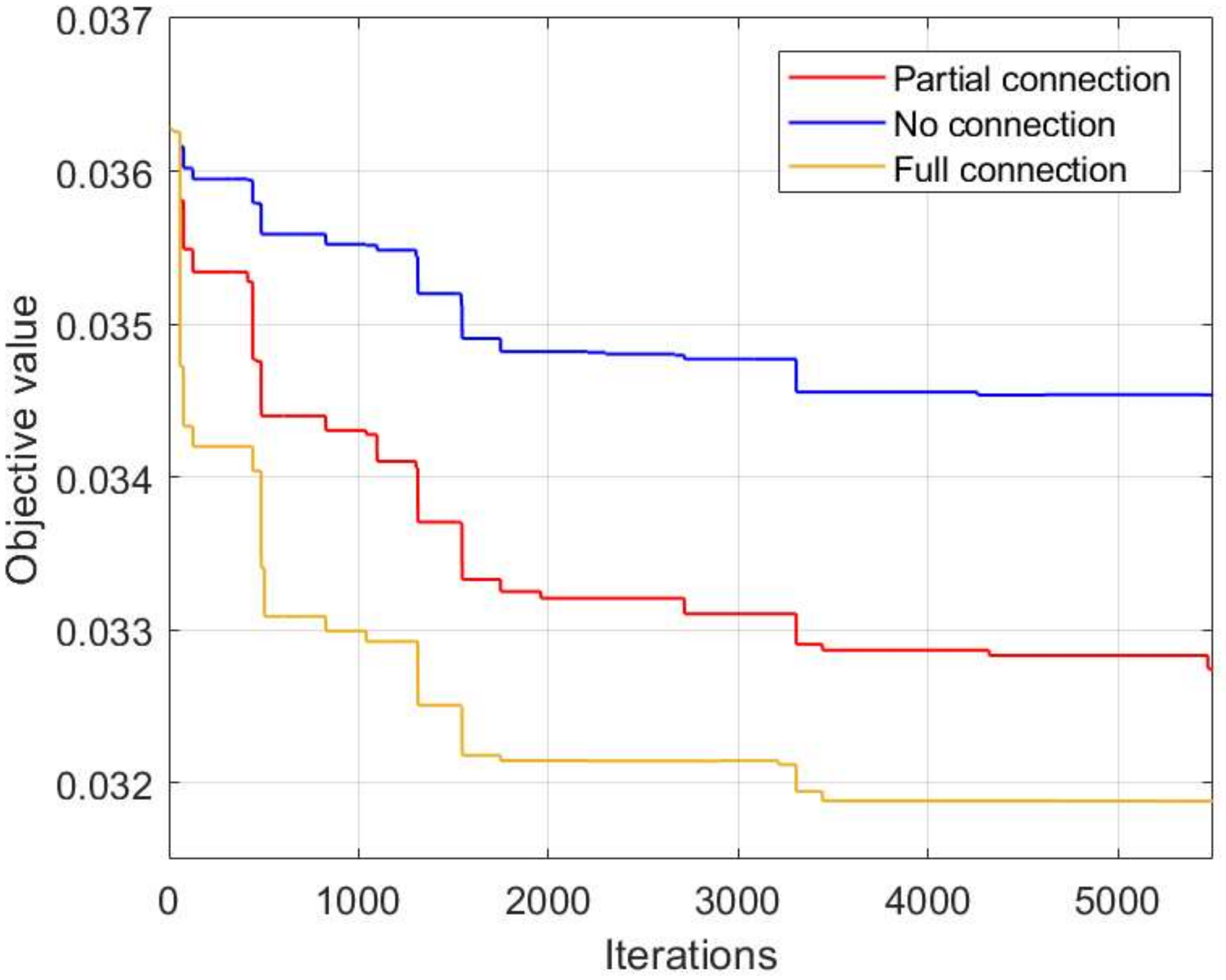}\\
  \caption{Impact of edge node connectivity: the smooth factor $\omega=10^{-6}$ and the weight factor $w_s=3\cdot 10^{-3}$.}\label{fig:impact of connection}
\end{figure}
This part analyzes the impact of edge node connectivity on the performance of ICE.
As shown in Fig. \ref{fig:impact of connection}, the system with all the edge nodes connected converges to the minimum objective value while the system with no edge nodes connected has the highest objective value.
In the system with all the edge nodes connected, the benefits of cooperation can be achieve at system level through scheduling workloads among all the edge nodes.
When edge nodes are partially connected, the cooperation benefits can only be explored within clusters (the edge nodes within a cluster are connected and different clusters are not connected with each other).
Therefore, the higher extent that edge nodes are connected with each other, the more cooperation benefits can be achieved by ICE.
\section{Conclusions}
In this paper, we have investigated cooperative service caching and workload scheduling in mobile edge computing.
Based on queuing analysis, we have formulated this problem as a mixed integer nonlinear programming problem, which is proved to have non-polynomial computation complexity.
To deal with the challenges of subproblem coupling, computation-communication tradeoff and edge node heterogeneity, we have proposed ICE based on Gibbs sampling to achieve the near-optimal service caching policy in an iterative manner.
We have further presented a water-filling based workload scheduling algorithm, which has polynomial computation complexity.
Extensive simulations have been conducted to evaluate the effectiveness and convergence of the proposed algorithm, and the impact of edge connectivity is further analyzed.
\begin{appendices}
\section{Proof of Theorem \ref{theorem:convengence property}}
\label{proof:proof of convergence property}
Let ${\bm{A}} = \{ \bm{a}_1,\bm{a}_1,...,\bm{a}_M\} $ be the caching decision space of edge nodes, and in each iteration, a random edge node $n$ randomly chooses a caching decision from ${\bm A}$.
 With Algorithm 1 iterating over the edge nodes and the caching decision space, the \emph{edge caching policy} $\bm{C}$ evolves as a $N$-dimension Markov chain, in which each dimension represents the caching decision of each edge node.
For the convenience of presentation, we analyze the scenario with 2 edge nodes, and the 2-dimension Markov chain is denoted as $\left\langle {\bm{c_1},\bm{c_2}} \right\rangle $.
In each iteration, one randomly selected edge node $n$ ($n\in\mathds{N}$) virtually changes its  current caching decision to a random caching decision from $\bm{c_n}$, thus there is
\begin{equation}\label{equ: transition probability}
  \begin{array}{l}
\Pr (\left\langle {c_1^*,{c_2}} \right\rangle |\left\langle {{c_1},{c_2}} \right\rangle ) = \frac{{{e^{ - \frac{{y(\left\langle {c_1^*,{c_2}} \right\rangle )}}{\omega }}}}}{{NM({e^{ - \frac{{y(\left\langle {c_1^*,{c_2}} \right\rangle )}}{\omega }}} + {e^{ - \frac{{y(\left\langle {{c_1},{c_2}} \right\rangle )}}{\omega }}})}}\\
\Pr (\left\langle {{c_1},c_2^*} \right\rangle |\left\langle {{c_1},{c_2}} \right\rangle ) = \frac{{{e^{ - \frac{{y(\left\langle {{c_1},c_2^*} \right\rangle )}}{\omega }}}}}{{NM({e^{ - \frac{{y(\left\langle {{c_1},c_2^*} \right\rangle )}}{\omega }}} + {e^{ - \frac{{y(\left\langle {{c_1},{c_2}} \right\rangle )}}{\omega }}})}},
\end{array}
\end{equation}
where $y(\left\langle {{c_1},{c_2}} \right\rangle )$ is the objective value when the caching policy is $\left\langle {{c_1},{c_2}} \right\rangle $.
In this scenario, $N=2$.
Denote by $\pi (\left\langle {{c_1},{c_2}} \right\rangle )$ the stationary probability distribution of caching policy $\left\langle {{c_1},{c_2}} \right\rangle $, then $\pi (\left\langle {{c_1},{c_2}} \right\rangle )$ can be derived by the fine stationary condition of the Markov chain as
\begin{equation}\label{equa: stationary condition}
\begin{array}{l}
  \pi(\left\langle {{a_1},{a_1}} \right\rangle )\Pr (\left\langle {{a_1},{a_m}} \right\rangle |\left\langle {{a_1},{a_1}} \right\rangle )\\
  =\pi(\left\langle {{a_1},{a_m}} \right\rangle )\Pr (\left\langle {{a_1},{a_1}} \right\rangle |\left\langle {{a_1},{a_m}} \right\rangle ).
  \end{array}
\end{equation}
Substitute (\ref{equ: transition probability}) into (\ref{equa: stationary condition}), it can be derived that
\begin{equation}\label{equa:symmetry}
\begin{array}{l}
  \pi (\left\langle {{a_1},{a_1}} \right\rangle ) \times \frac{{{e^{ - \frac{{y(\left\langle {{a_1},{a_m}} \right\rangle )}}{\omega }}}}}{{NM({e^{ - \frac{{y(\left\langle {{a_1},{a_m}} \right\rangle )}}{\omega }}} + {e^{ - \frac{{y(\left\langle {{a_1},{a_1}} \right\rangle )}}{\omega }}})}}\\=
  \pi (\left\langle {{a_1},{a_m}} \right\rangle ) \times \frac{{{e^{ - \frac{{y(\left\langle {{a_1},{a_1}} \right\rangle )}}{\omega }}}}}{{NM({e^{ - \frac{{y(\left\langle {{a_1},{a_m}} \right\rangle )}}{\omega }}} + {e^{ - \frac{{y(\left\langle {{a_1},{a_1}} \right\rangle )}}{\omega }}})}}.
  \end{array}
\end{equation}
It can be observed that Eq. (\ref{equa:symmetry}) is symmetric and can be balanced if $\pi (\left\langle {{c_1},{c_2}} \right\rangle )$ has the form of $\pi (\left\langle {{c_1},{c_1}} \right\rangle )=\gamma {{e^{ - \frac{{y(\left\langle {{c_1},{c_2}} \right\rangle )}}{\omega }}}}$, where $\gamma$ is a constant.
Let $\Phi$ be the caching policy space.
To ensure that $\sum\limits_{\left\langle {{c_1},{c_2}} \right\rangle  \in \Phi } {\pi (\left\langle {{c_1},{c_2}} \right\rangle )}  = 1$, the stationary probability distribution ${\pi (\left\langle {{c_1},{c_2}} \right\rangle )}$ should be given as
\begin{equation}\label{equa:stationary distribution}
\pi (\left\langle {{c_1},{c_2}} \right\rangle ) = \frac{{{e^{ - \frac{{y(\left\langle {{c_1},{c_2}} \right\rangle )}}{\omega }}}}}{{\sum\limits_{\left\langle {c_1^f,c_2^f} \right\rangle  \in \Phi } {{e^{ - \frac{{y(\left\langle {c_1^f,c_2^f} \right\rangle )}}{\omega }}}} }}
\end{equation}
Eq. (\ref{equa:stationary distribution}) can be rewritten as
\begin{equation}\label{equa:rewritten}
  \pi (\left\langle {{c_1},{c_2}} \right\rangle ) = \frac{1}{{\sum\limits_{\left\langle {c_1^f,c_2^f} \right\rangle  \in \Phi } {{e^{\frac{{y(\left\langle {{c_1},{c_2}} \right\rangle ) - y(\left\langle {c_1^f,c_2^f} \right\rangle )}}{\omega }}}} }}.
\end{equation}
Let $\left\langle {{c_1^*},{c_2^*}} \right\rangle$ be the globally optimal solution that minimizes the objective value, i.e., $y(\left\langle {{c_1^*},{c_2^*}} \right\rangle)\le y(\left\langle {{c_1^f},{c_2^f}} \right\rangle)$ for any $\left\langle {{c_1^f},{c_2^f}} \right\rangle \in\Phi$.
It can be concluded that $\pi (\left\langle {{c_1^*},{c_2^*}} \right\rangle )$ increases as $\omega$ decreases, and $\pi (\left\langle {{c_1^*},{c_2^*}} \right\rangle )\rightarrow1$ when $\omega\rightarrow0$.
\section{Proof of Theorem \ref{theorem:convexity property}}
\label{proof:proof of convexity}
An optimization problem should satisfy that the objective function and the inequation constraint functions are convex, and the equation constraint function is affine over the decision variables.
It is easy to identify that the inequation and equation constraint functions satisfy these conditions.
We just need to prove the convexity of the objective function.

In Eq. (\ref{equa: objective function}), it is intuitive that $\sum\limits_{s \in S} {({\lambda _{os}}{d_{{\rm{cloud}}}} + {w_s}{\lambda _{os}}} )$ and $\sum\limits_{s \in S} {\sum\limits_{n \in N} {\frac{{\max \{ {\lambda _{ns}}{A_s} - {A_{ns}},0\} }}{{{A_s}}}{d_n}}}$ are convex over $\bm{\Lambda}$.
Let $x(\bm{\Lambda})= \sum\limits_{s \in S} {\sum\limits_{n \in N} {\frac{{{\lambda _{ns}}}}{{{\mu _{ns}} - {\lambda _{ns}}{A_s}}}}}$.
Denote by $H=[h_{mn}]_{m\times n}$ the Hessian matrix of $x(\bm{\Lambda})$, and $h_{mn}$ ($m\in \mathds{N}$, $n\in\mathds{N}$) can be given as
\begin{equation}\label{equa:Hessian marix}
  {h_{mn}} = \left\{ \begin{aligned}
\frac{{2{\mu _{ns}}{A_s}}}{{{{({\mu _{ns}} - {A_s}{\lambda _{ns}})}^3}}},~~~~~&\rm{if}~\emph{m=n}\\
0.~~~~~~~~~~~~~&{\rm{otherwise}}
\end{aligned} \right.
\end{equation}
In Eq. (\ref{equa:Hessian marix}), $h_{mn}>0$ if $m=n$, and otherwise, $h_{mn}=0$.
Therefore, $H$ is a positive definite matrix, and $x(\bm{\Lambda})$ is convex over $\bm{\Lambda}$ \cite{boyd2004convex}.
The objective function $f(\bm{\Lambda})$ is the sum of several convex functions over $\bm{\Lambda}$, so $f(\bm{\Lambda})$ is also convex over $\bm{\Lambda}$.
Thus, we can conclude that problem \textbf{P2} is a convex optimization problem over the workload scheduling policy $\bm{\Lambda}$.
\end{appendices}

%

\bibliographystyle{IEEEtran}

\end{document}